\pdfoutput=1 
\documentclass[letterpaper, 10 pt, conference, onecolumn]{ieeeconf} 

\IEEEoverridecommandlockouts                              
\usepackage{amsmath}
\usepackage{xcolor}
\usepackage{amssymb}
\usepackage{algorithm}
\usepackage{graphicx}
\usepackage{algpseudocode}
\usepackage{cite} 

\usepackage{amsthm}   

\theoremstyle{definition} 
\newtheorem{definition}{Definition}
\newtheorem{remark}{Remark}

\newtheorem{assumption}{Assumption}
\newtheorem{problem}{Problem}

\theoremstyle{plain}      
\newtheorem{theorem}{Theorem}

\newtheorem{property}{Property}

\title{\LARGE \bf
Robust Data-Driven Invariant Sets for Nonlinear Systems
}

\author{Sahand Kiani and Constantino M. Lagoa
\thanks{The authors are with the Department of Electrical and Computer Engineering, Pennsylvania State University, State College, PA 16801, USA. S. Kiani is a PhD student (e-mail: szk6437@psu.edu). C. M. Lagoa is a full professor (e-mail: cml18@psu.edu).}%
}

\begin{document}

\maketitle
\thispagestyle{empty}
\pagestyle{empty}

\begin{abstract}
The synthesis of robust invariant sets for nonlinear systems has traditionally been hindered by the inherent non-convexity and a strict reliance on exact analytical models. This paper presents a purely data-driven framework to compute robust polytopic contractive sets for unknown nonlinear systems operating under persistent bounded process noise and state-input constraints. Rather than attempting to identify a single, potentially nominal model, we utilize a finite data set to construct a polytopic consistency set—a rigorous geometric boundary encapsulating all possible system dynamics compatible with the noisy measurements. The core contribution of this work extends an established sufficient condition for $\lambda$-contractiveness into the data-driven setting. Crucially, we prove that enforcing this condition strictly over the vertices of the consistency set  guarantees robust invariance.
\end{abstract}

\section{INTRODUCTION}

One of the aspects in the study of dynamical systems is the analysis of their invariant sets \cite{blanchini1999set}. The task of computing such a set is of significant practical importance, as these sets not only guarantee invariance but are also associated with the requirements of safety and stability of systems in many applications \cite{mayne2000constrained,kiani2024learning}. This property makes these sets essential for designing high-performance controllers, particularly for their use as terminal sets in Model Predictive Control to ensure recursive feasibility and stability. Approximating invariant sets for nonlinear systems is significantly more complex than for their linear counterparts, primarily due to the non-convex nature of the resulting sets and the requirement for robustness against uncertainty. These challenges are magnified in the data-driven setting, where the lack of an explicit model invalidates most traditional analysis tools.

\subsection{Related Work}

The computation of invariant sets for nonlinear systems has been addressed using various optimization-based techniques. For instance, the authors in \cite{10123035} propose a method for perturbed nonlinear sampled-data systems by formulating the problem as a non-convex program. They introduce a tailored successive convexification algorithm that leverages a zonotopic representation for the sets and is designed to find a large region of safe operation by ensuring the sequence of candidate sets is monotonically increasing in volume. Other approaches for handling constrained nonlinear systems focus on approximating the one-step reachable set. For example, the work in \cite{bravo2005computation} presents a method to compute a guaranteed inner approximation of the set for known nonlinear models. The authors develop a specialized branch and bound algorithm based on interval arithmetic, which is followed by a second procedure to find a simple inner polytopic representation of the computed region.

An alternative line of research approaches the computation of invariant sets using graph-based methods. The authors in \cite{decardi2021computing}, for example, cast the problem of finding the largest robust control invariant sets for constrained nonlinear systems as a problem in graph theory. Their methodology involves approximating the system dynamics with a directed graph and then applying an iterative subdivision algorithm to generate a sequence of outer approximations that is proven to converge to the true largest robust control invariant set.

Existing methods for computing robust invariant sets, such as successive convexification, interval arithmetic, and graph-based discretization, are fundamentally limited by their reliance on an explicit mathematical model. This limitation motivates the purely data-driven framework proposed in this paper.
\subsection{Key Contribution}

The primary contribution of this paper is the development of a data-driven framework for approximating a robust convex polytopic contractive set for a class of noisy nonlinear systems. We consider a class of uncertain nonlinear systems characterized by known polynomial basis functions with unknown parameters. The system is subject to persistent, bounded process noise as well as both state and input constraints. The core theoretical contribution is the extension of an established sufficient condition for contractiveness—originally developed for known nonlinear models subject to noise—to the data-driven setting. We achieve robustness by requiring that the computed set be invariant for all system dynamics compatible with the finite collected noisy measurements. Finally, leveraging the properties of DC functions, we present a tractable convex optimization algorithm that connects this condition to compute the largest possible $\lambda$-contractive set within the state constraints. 

\textbf{Notation.}
The set of real numbers is denoted by $\mathbb{R}$, and the $n$-dimensional Euclidean space is $\mathbb{R}^n$. The state vector at time step $k$ is denoted by $x_k \in \mathbb{R}^n$, with its $j$-th component denoted by $x_{j,k}$. The operator $\nabla_x$ denotes the gradient with respect to the vector $x$. The infinity norm of a vector is denoted by $\|\cdot\|_{\infty}$. For every $n \in \mathbb{N}$, define $\mathbb{N}_n = \{i \in \mathbb{N} : 1 \le i \le n\}$. The vector of all ones with the proper dimension is denoted by $\mathbf{1}$. The interior of a set is denoted by $\text{int}(\cdot)$. The notation $\lambda\Omega$ represents the scaling of the set $\Omega$ by the constant scalar $\lambda$. Consider the matrix $H \in \mathbb{R}^{n \times n}$, the notation \textbf{$H_i$} represents the $i$-th row of this matrix. Throughout this paper, the notation $\theta = \text{vec}(\Theta^T)$ defines the vector $\theta$ as the vectorization of the transpose of the matrix $\Theta$. Given a set $P$, co($P$) denotes its convex hull.

\section{Preliminary Results}\label{sec:preliminaries}

In this paper, we consider the following discrete-time nonlinear system:
\begin{equation} \label{eq:system}
x_{k+1} = f(x_k, u_k) + w_k,
\end{equation}
where $x_k \in X \subseteq \mathbb{R}^n$ is the state of the system, $u_k \in U \subseteq \mathbb{R}^m$ is the control input, and $w_k \in W \subset \mathbb{R}^n$ is an unknown but bounded uncertainty such that $\|w\|_\infty \leq \epsilon$ where $\epsilon$ is assumed to be known.

\begin{assumption}
The sets $X$, $U$, and $W$ are assumed to be compact, convex, and to contain the origin in their interiors.\label{assump1}
\end{assumption}

\subsection{Difference of Convex Functions}

A key advantage of the DC framework is its generality, as a wide class of nonlinear functions can be represented as the Difference of Convex (DC) functions. This structure is particularly useful because it allows for the application of powerful tools from convex analysis to otherwise non-convex problems. Therefore, we assume the following for the system (\ref{eq:system}):

\begin{assumption}
The system dynamics vector, represented by $f(x,u): X \times U \to \mathbb{R}^n$, is assumed to be a vector of DC functions, differentiable at $(0,0)$. This implies there exist two functions $g_j, h_j: X \times U \to \mathbb{R}$ for all $j \in \mathbb{N}_n$ of vector function $f(x,u)$, which are convex $\forall (x_k, u_k) \in X \times U$, that satisfy the following conditions:
\begin{align*}
    &f_j(x_k, u_k) = g_j(x_k, u_k) - h_j(x_k, u_k),  \forall (x_k, u_k) \in X \times U,  \\
    &g_j(0, 0) = h_j(0, 0) = 0, \quad \forall j \in \mathbb{N}_n,
\end{align*}\label{assump2}
\end{assumption}
\vspace{-25pt}
\begin{remark}
The class of DC functions is quite large: Every function with continuous second partial derivatives is a DC function on a convex compact set of its domain \cite{horst1999dc}. Consequently, the requirement in Assumption \ref{assump2} that the system dynamics need to be a DC function is not a significant restriction, ensuring that the methods proposed in this paper are applicable to a broad class of nonlinear systems.
\end{remark}

With this functional framework established, we now introduce the key definitions of $\lambda$-contractiveness and sufficient condition for a set to be contractive for nonlinear systems that are fundamental to our proposed method.

\subsection{$\lambda$-Contractiveness for Nonlinear Systems}

Before presenting our main results, it is necessary to review the key concepts from set-theoretic control upon which our work is built. Therefore, we start with the following definitions:

\begin{definition}
A set $\Omega \subseteq X$ is a robust control invariant set for the system $x_{k+1} = f(x_k, u_k) + w_k$ subject to constraints $x_k \in X$ and $u_k \in U$ if, for every $x_k \in \Omega$, there exists a control input $u_k \in U$ such that $f(x_k, u_k) + w_k \in \Omega$ for all $w_k \in W$.
\end{definition}
 
\begin{definition}\cite{272351}: 
    Given the contraction factor $\lambda \in [0, 1]$, a convex compact set $\Omega$ with $0 \in \text{int}(\Omega)$ is (robust) $\lambda$-contractive for system $x_{k+1} = f(x_k, u_k) + w_k$ if for any $x_k \in \Omega \subseteq X$, there exists an input $u_k \in U$ such that $f(x_k, u_k) + w_k \in \lambda\Omega$, $\forall w_k \in W$.
\end{definition} 

The definition of a $\lambda$-contractive set is a generalization of an invariant set. In the specific case where the contraction factor is set to $\lambda = 1$, the condition $f(x_k, u_k) + w_k \in \lambda\Omega$ becomes identical to the condition for invariant sets.

\begin{remark}
In the specific case where the set $\Omega$ is a polytope described by a system of linear inequalities, i.e., $\Omega = \{x \in \mathbb{R}^n : Hx \le \textbf{1}\}$, the condition for $\lambda$-contractiveness can be stated as follows: For every $x_k \in \Omega \subseteq X$, there exists a control input $u_k \in U$ such that $\forall w_k \in W$, the resulting $f(x_k, u_k) + w_k$ satisfies the polytope's defining inequalities scaled by $\lambda$: $H(f(x_k, u_k) + w_k) \le \lambda \mathbf{1}$.
\end{remark}

Building on the definition of a robust control invariant set, we now review a key result from \cite{fiacchini2010computation}. The main contribution of that work was to utilize the properties of DC functions to establish a sufficient condition for a set to be $\lambda$-contractive for nonlinear systems.

\begin{definition}\cite{boyd2004convex}: 
Given a set $\Gamma \subseteq \mathbb{R}^n$, the support function of $\Gamma$ evaluated at $c \in \mathbb{R}^n$ is defined as: $\phi_{\Gamma}(c) = \sup_{x \in \Gamma} c^T x$.    
\end{definition}

\begin{definition}\cite{fiacchini2010computation}: 
Let Assumptions \ref{assump1} and \ref{assump2} hold. Given the vector of DC function $f(\cdot, \cdot) : X \times U \to \mathbb{R}^n$, whose components are decomposed as $f_j(x,u) = g_j(x,u) - h_j(x,u), \forall j \in \mathbb{N}_n$ with convex functions $g_j(x,u)$ and $h_j(x,u)$, and $c \in \mathbb{R}^n$, define $F(x, u, c) : X \times U \times \mathbb{R}^n \to \mathbb{R}$ as the function
\begin{equation}
\begin{aligned}
F(x, u, c) = &\sum_{j \in b_+} c_j(g_j(x, u) - h_j^L(x, u)) \\
&+ \sum_{j \in b_-} c_j(g_j^L(x, u) - h_j(x, u)),
\end{aligned}\label{eq:F}
\end{equation}
\end{definition} 
\hspace{-10pt}where $g_j^L(x, u) = \nabla_x g_j(0, 0)x + \nabla_u g_j(0, 0)u$ and $h_j^L(x, u) = \nabla_x h_j(0, 0)x + \nabla_u h_j(0, 0)u$, for $j \in \mathbb{N}_n$, and $b_+ = b_+(c) = \{j \in \mathbb{N}_n : c_j \ge 0\}$ and $b_- = b_-(c) = \{j \in \mathbb{N}_n : c_j < 0\}$.

The function $F(x, u, c)$ exhibits the following important property:

\begin{property}
 \cite{fiacchini2010computation}: Let Assumptions 1 and 2 hold. Given the DC function $f(\cdot, \cdot) : X \times U \to \mathbb{R}^n$ as in (1), then for every $c \in \mathbb{R}^n$, the function $F(\cdot, \cdot, c)$ defined in (2) is convex in $(x, u) \in X \times U$. Moreover, it can be shown that for any $c \in \mathbb{R}^n$, the inequality $c^T f(x,u) \leq F(x, u, c)$ holds $\forall (x,u) \in X \times U$.\label{propconv}  
\end{property}

Now the result will be used to provide a sufficient condition for robust control invariance of a polytope for the uncertain nonlinear system (1). In what follows, given a polytope $\Omega = \{x \in \mathbb{R}^n : Hx \le \mathbf{1}\} \subseteq X$, its $n_v$ vertices are denoted $v^p \in \mathbb{R}^n$, for $p \in \mathbb{N}_{n_v}$, and $n_h$ are the rows of $H$, i.e. $H \in \mathbb{R}^{n_h \times n}$.

The following property provides the necessary and sufficient condition for $\lambda$-contractiveness (and then for robust control invariance) for system (1) as follows:

\begin{property} \cite{fiacchini2010computation}:
Let Assumptions \ref{assump1} and \ref{assump2} hold. Consider a polytope $\Omega = \{x \in \mathbb{R}^n : Hx \le \mathbf{1}\} \subseteq X$, and DC system (1) which is subject to process noise. If there exist control actions defined at the vertices $u^p = u_v^p \in U$, $\forall p \in \mathbb{N}_{n_v}$, such that
\begin{equation}
    F(v^p, u_v^p, H_i^T) \le \lambda_\omega - \phi_W(H_i^T), \quad \forall p \in \mathbb{N}_{n_v}, \forall i \in \mathbb{N}_{n_h}.\label{condi}
\end{equation}
where $\lambda_\omega \in [0, 1]$ represents the contraction factor for noisy nonlinear system, then $\Omega$ is a $\lambda_\omega$-contractive set for system (1) and constraints $x \in X, u \in U$.
\end{property}\label{prop:robust_contractiveness}


\section{Problem Statement}\label{sec:problem_statement}

This section formulates the data-driven synthesis of robust $\lambda$-contractive sets for the unknown nonlinear system \eqref{eq:system}. We introduce the following standing assumptions to ensure the problem is well-posed and computationally tractable:

\begin{assumption} \label{assum:model_structure}
The nonlinear function $f(x_k,u_k)$ is assumed to lie in the span of a known dictionary of $d_f$ basis functions, $\phi(x,u): X \times U \to \mathbb{R}^{d_f}$. Consequently, there exists an unknown parameter matrix $\Theta \in \mathbb{R}^{n \times d_f}$ such that the system dynamics can be expressed as $f(x_k,u_k) = \Theta\phi(x_k,u_k)$. \label{assump3}
\end{assumption}

Under these assumptions, the problem of computing the maximal robust invariant set is defined as:

\begin{problem} \label{prob:main_problem_revised}
Under Assumptions (\ref{assump1}--\ref{assump3}), the objective is to compute the largest polytopic robust $\lambda$-contractive set, $\Omega \subseteq X$, using the finite dataset $\mathcal{D}=\left\{x_k, u_k, x_{k+1}\}_{k=0}^{T-1}\right.$ that is generated from the noisy nonlinear system (\ref{eq:system}). The computed set must be robust $\lambda$-contractive for all possible system dynamics that are consistent with the data dictionary, $\mathcal{D}$, and the known noise upper bound $\epsilon$.
\end{problem}

\section{Data-Driven Computation of Robust Contractive Set}\label{sec:main_results}

In this section, we address Problem \ref{prob:main_problem_revised} by developing a constructive framework for robust set synthesis. Our approach relies on a set of feasible system parameters, referred to as the consistency set, defined as follows:

\begin{definition}
The feasible parameter set of all systems compatible with the data dictionary $\mathcal{D}=\left\{x_k, u_k, x_{k+1}\}_{k=0}^{T-1}\right.$ and the known noise upper bound $\epsilon$ is the consistency set which is defined as follows
\begin{equation}\label{eq:cons}
\begin{aligned}
    \mathcal{P}_1 \doteq \Big\{ \Theta \in \mathbb{R}^{n \times d_f} : \;& \|x_{k+1} - \Theta\phi(x_k,u_k)\|_{\infty} \le \epsilon, \\
    & \forall k \in \{0, \dots, T-1\} \Big\}.
\end{aligned}
\end{equation}
\end{definition}

The consistency set defined in \eqref{eq:cons} is a polytope and can be expressed in the following equivalent form:

\begin{equation}
    \mathcal{P}_1 = \left\{\theta : \begin{bmatrix} A \\ -A \end{bmatrix} \theta \le \begin{bmatrix} \epsilon \mathbf{1} + \xi \\ \epsilon \mathbf{1} - \xi \end{bmatrix} \right\}, \label{eq:consisten}
\end{equation}
where $\theta = \text{vec}(\Theta^T)$ and $\epsilon$ are the unknown parameters and known noise upper bound, respectively. The matrices $A, \xi$ are functions of the collected data:
\begin{equation}
    A \doteq \begin{bmatrix} I \otimes \phi^T(x_0,u_0) \\ \vdots \\ I \otimes \phi^T(x_{T-1}, u_{T-1}) \end{bmatrix}, \quad \xi \doteq \begin{bmatrix} x_1 \\ \vdots \\ x_T \end{bmatrix},
\end{equation}

The main reason for expressing the consistency set in this polytopic form is computational. The polytopic representation of the consistency set for linear systems is investigated in \cite{zheng2025robust}. This representation allows the set of all possible system parameters to be finitely characterized by its vertices.

\begin{assumption} \label{assump4}
Enough data, in the form of a dataset $\mathcal{D} = \{(x_k, u_k, x_{k+1})\}_{k=0}^{T-1}$, has been collected from trajectories of the system \eqref{eq:system} such that the data matrix $A$ has full column rank. This ensures that the consistency set $\mathcal{P}_1$ is compact.
\end{assumption}

\begin{remark}
Assumption \ref{assump4} is critical for the proposed framework. The full column rank of matrix $A$ guarantees that the consistency set $\mathcal{P}_1$ is compact, and therefore has a finite diameter. If this condition were not met, the set of possible system parameters would be unbounded.
\end{remark}

By Assumption \ref{assump4}, the consistency set is defined by the convex hull of its vertices. These vertices can be computed using standard vertex enumeration algorithms\cite{avis1991pivoting}. We denote the set of all computed vertices of the consistency set to be $Q = \{\theta^{*,1}, \theta^{*,2}, \dots, \theta^{*,n_q}\}$ that contains $n_q$ vertices in total.

We now explicitly characterize the dependence of the function F in \eqref{eq:F} on the system parameters through the following relationships: Each computed vertex of the consistency set represents a distinct possible model parameter for the noisy nonlinear system. For each vertex $s \in \mathbb{N}_{n_q}$, we define $\Theta^{*,s}$ as the matrix representation of $\theta^{*,s}$. The associated dynamic is:
\begin{equation}
     f(x_k,u_k; \theta^{*,s}) = \Theta^{*,s}\phi(x_k,u_k), \forall 
     s\in \mathbb{N}_{n_q}, 
\end{equation}

Since Assumption~\ref{assump2} holds, the convex components of $f$, $g$ and $h$, inherit this dependency, such that for each vertex we have the decomposition:
\begin{equation}
    \begin{aligned}
        &f(x_k,u_k; \theta^{*,s}) = g(x_k,u_k ; \theta^{*,s}) - h(x_k,u_k; \theta^{*,s}), \forall s \in \mathbb{N}_{n_q}
    \end{aligned}
\end{equation}
where $g(\cdot ; \theta^{*,s}),h(\cdot ; \theta^{*,s})$ represents vector of convex function associated with $s^{th}$ component of the set $Q$.

\begin{remark}
Because the basis functions $\phi$ are known, decomposing $f$ into convex components $g$ and $h$ is computationally tractable using established methods like the DC Sums-of-Squares (DC-SOS) framework \cite{niu2018difference}, which leverages spectral decomposition and recursive monomial factorization.
\end{remark}

Consequently, the function $F$, which is constructed from these DC components and their linearizations, is also parameterized by the specific vertex $\theta^{*,s} \in Q$ under consideration. This leads to the following formal expression:

\begin{definition} \label{def:F_parameterized}
For any vector $c \in \mathbb{R}^n$, we define the function $F(x, u, c; \theta^{*,s})$ as:
\begin{equation} \label{eq:F_theta}
\begin{aligned}
    F(x, u, c; \theta^{*,s}) = &\sum_{j \in b_+} c_j \left( g_j(x, u; \theta^{*,s}) - h_j^L(x, u; \theta^{*,s}) \right) \\
    &+ \sum_{j \in b_-} c_j \left( g_j^L(x, u; \theta^{*,s}) - h_j(x, u; \theta^{*,s}) \right),
\end{aligned}
\end{equation}
where $g_j(\cdot, \cdot; \theta^{*,s})$ and $h_j(\cdot, \cdot; \theta^{*,s})$ are the $j^{th}$ convex components of the system dynamics corresponding to the vertex $\theta^{*,s}$ where $s \in \mathbb{N}_{n_q}$, and $g_j^L$, $h_j^L$ are their respective linearizations at the origin. The index sets are defined as $b_+ = b_+(c) = \{j \in \mathbb{N}_n : c_j \ge 0\}$ and $b_- = b_-(c) = \{j \in \mathbb{N}_n : c_j < 0\}$
\end{definition}

Having redefined the function $F$ to account for the parameter uncertainty, we now propose the following sufficient condition for a set to be robust $\lambda_\omega$-contractive for the unknown noisy nonlinear system.

\begin{theorem}\label{theo}
Let Assumptions \ref{assump1}--\ref{assump4} hold. Consider a polytope $\Omega = \{x \in \mathbb{R}^n : Hx \le \mathbf{1}\} \subseteq X$. If there exist control actions $u^p = u^p_v \in U$ defined for each vertex $v^p$ of $\Omega$, such that
\begin{equation}
    F(v^p, u^p_v, H_i^T; \theta^{*,s}) \le \lambda_\omega - \phi_W(H_i^T), \forall s \in \mathbb{N}_{n_q},
\end{equation}
holds for all polytope facets $i \in \mathbb{N}_{n_h}$, then for a given $\lambda_\omega \in [0, 1]$, the polytope $\Omega$ is a $\lambda_\omega$-contractive set for all systems compatible with data dictionary $\mathcal{D}$.
\end{theorem}

\begin{proof}
    Please check the Appendix.
\end{proof}
Based on Theorem \ref{theo}, we now present Algorithm 1. This optimization-based procedure is designed to compute the largest possible polytopic $\lambda_\omega$-contractive set within the state constraints for the unknown noisy nonlinear system.

\begin{algorithm}[t]
\caption{Computation of a $\lambda_\omega$-Contractive Set}
\label{alg:data_driven_compute_set}
\begin{algorithmic}[1]
    \State \textbf{Given} $\Omega_0$, $Q = \{\theta^{*,1}, \dots, \theta^{*,n_q}\}$, $\lambda_\omega$, 
    \For{$s = 1, \dots, n_q$}
        \For{$p = 1, \dots, n_v$}
            \State solve the following for $\alpha^{s,p}$ and $\forall i \in \mathbb{N}_{n_h}$:
            \State \quad $\alpha^{s,p} = \max\limits_{\gamma^{p} > 0, u^{p} \in U} \gamma^{p}$
            \State \text{s.t. } $F(\gamma^{p} v^{p}, u^{p}, H_i^T; \theta^{*,s}) \le \lambda_\omega \gamma^{p} - \phi_W(H_i^T)$
        \EndFor
    \EndFor
    \State $\alpha = \min_{s \in \{1,\dots,n_q\}, p \in \{1,\dots,n_v\}} \{\alpha^{s,p}\}$
    \State \textbf{return} $\alpha\Omega_0$
\end{algorithmic}
\end{algorithm}

It is worth mentioning that the variable $\gamma^p$ acts as a scalar multiplier determining the maximum permissible scaling of each candidate vertex $v^p$ along its original direction.
\begin{remark}
The computational complexity involves two stages: First, enumerating the $n_q$ vertices of the consistency set from $2T$ inequalities scales exponentially with the parameter dimension. Second, the optimization in Algorithm \ref{alg:data_driven_compute_set} scales proportionally to $n_q \times n_v \times n_h$. Therefore, selecting the initial candidate polytope $\Omega_0$ presents a clear trade-off: increasing its complexity yields a less conservative $\lambda$-contractive set, but directly increases the computational load.
\end{remark}

\begin{remark}
The DC decomposition ($f = g - h$) relaxes the intractable non-convex invariance condition into a tractable convex optimization problem by constructing a convex upper bound, $F$. Crucially, this decomposition is not unique. This flexibility can be exploited to reduce conservatism: by evaluating multiple valid decompositions for a fixed initial polytope $\Omega_0$, one can select the decomposition that maximizes the scaling factor $\alpha$, resulting in the least conservative $\lambda$-contractive set.
\end{remark}
\subsection{Enlarging $\lambda$-Contractive Set}

The method presented in Algorithm \ref{alg:data_driven_compute_set} is effective at finding the largest $\lambda$-contractive set for a \textit{given} initial polytope shape. However, the final result is inherently limited by this initial choice and may be a conservative approximation of the true maximal invariant set. To overcome this limitation and find a larger, less conservative set, an iterative enlarging method can be applied. Such a method systematically expands a known invariant set by adding new, provably safe points within the state constraint set. In this subsection, we review an enlarging procedure based on the work of \cite{fiacchini2010computation}.

\begin{remark}
    \cite{fiacchini2010computation}:
Let Assumptions \ref{assump1} and \ref{assump2} hold. Consider a polytope $\Omega = \{x \in \mathbb{R}^n : Hx \le \mathbf{1}\} \subseteq X$, with $H \in \mathbb{R}^{n_h \times n}$, and $\lambda_\omega \in [0, 1]$, such that Theorem \ref{theo} holds for $\Omega$, and, given a new point $\hat{x} \in X$, define the set $\hat{\Omega} = \text{co}(\Omega \cup \hat{x})$. If there exists an input $\hat{u} \in U$ such that $F(\hat{x}, \hat{u}, H_i^T; \theta^{*,s}) \le \lambda_\omega - \phi_W(H_i^T)$, for every $i \in \mathbb{N}_{n_h}$ and $\forall s \in \mathbb{N}_{n_q}$, then $\hat{\Omega}$ is a robust $\lambda_\omega$-contractive set for system (1) and constraints $x \in X$ and $u \in U$.  \label{conv}  
\end{remark}

The authors in \cite{fiacchini2010computation} show that an existing $\lambda$-contractive set can be enlarged to reduce conservatism. For low-dimensional systems, they suggest that one can generate random candidate points in the state space and test if they can be added to the set as mentioned in Remark \ref{conv}. However, this approach becomes inefficient and non-trivial for high-dimensional systems. To address this, an alternative, more systematic procedure involves solving the following convex optimization problem to determine a point $\hat{x} \in X$ to enlarge $\Omega = \{x \in \mathbb{R}^n : Hx \le \mathbf{1}\}$. Given $c \in \mathbb{R}^n$
\begin{equation} \label{eq:enlarging_problem}
\begin{aligned}
    \max_{\hat{x} \in X, \hat{u} \in U} \quad & c^T \hat{x} \\
    \text{s.t.} \quad & F(\hat{x}, \hat{u}, H_i^T; \theta^{*,s}) \le \lambda_\omega - \phi_W(H_i^T), \\
    &i \in \mathbb{N}_{n_h}, \forall s \in \mathbb{N}_{n_q},
\end{aligned}
\end{equation}
where the enlarged $\lambda$-contractive set is $\hat{\Omega} = \text{co}(\Omega \cup \hat{x})$ where $\Omega \subseteq \hat{\Omega}$. In the following section, we will demonstrate the effectiveness of this approach with a numerical example.

\section{Simulation Results}\label{sec:simulations}

In this section, we illustrate the effectiveness of our proposed method to compute the largest possible invariant set within the state constraints for noisy nonlinear systems in a data-driven manner.

\textit{Case Study}: We consider the discretized system dynamics presented in \cite{cannon2003nonlinear} given by the equation:

\begin{align} 
    x_{k+1} &= \begin{bmatrix} 1 & \tau \\ \tau & 1 \end{bmatrix} x_k + \tau \left\{ \mu \begin{bmatrix} 1 \\ 1 \end{bmatrix} + (1-\mu) \begin{bmatrix} 1 & 0 \\ 0 & -4 \end{bmatrix} x_k \right\} u_k \label{testsys} \\
    &= \begin{bmatrix} a_1 & a_2 \\ a_4 & a_5 \end{bmatrix} x_k + \left\{ \begin{bmatrix} b_1 \\ b_2 \end{bmatrix} + \begin{bmatrix} a_3 & 0 \\ 0 & a_6 \end{bmatrix} x_k \right\} u_k, \label{eq:param_model}
\end{align}
where the parameters are $\tau=0.01$ and $\mu=0.9$. The parameters $\{a_1,a_2,\cdots,b_1,b_2\}$ are assumed unknown throughout the simulation. The constraints on the input and state are given by $U = \{u \in \mathbb{R} : |u| \le 2\}$ and $X = \{x \in \mathbb{R}^2 : \|x\|_{\infty} \le 4\}$, respectively. The uncertainty set for the continuous system is bounded by $W = \{w \in \mathbb{R}^{n} : \|w\|_{\infty} \le 0.4\}$, for which the upper bound is known. The data dictionary, $\mathcal{D} = \{(x_k, u_k, x_{k+1})\}_{k=0}^{k=30}$, is generated using (\ref{testsys}).

\textit{Simulation Environment}: The numerical simulations presented were carried out in Python/Matlab environment on an Apple M4 processor and 32 GB of RAM.

To construct the consistency set defined in \eqref{eq:consisten}, we first generate a dataset by sampling trajectories from the system \eqref{testsys}, collecting a total of 30 state-input pairs. Once the consistency set is constructed, its vertices are calculated using the \texttt{con2vert} function in MATLAB. For the given dataset size and noise bound, this computation yields $n_q=667$ vertices and takes approximately 1.3 seconds. Now, we need to construct the function $F$ as mentioned in (\ref{eq:F_theta}). To do this, one possible choice for the decomposition can be the following:


\begin{equation}\label{eq:decompo}
\begin{aligned}
    g_1(x_k, u_k) &= a_1 x_{1,k} + a_2 x_{2,k} + b_1 u_k \\ 
    &\quad + \frac{|a_3|}{4}\big(x_{1,k} + \text{sgn}(a_3)u_k\big)^2, \\
    h_1(x_k, u_k) &= \frac{|a_3|}{4}\big(x_{1,k} - \text{sgn}(a_3)u_k\big)^2, \\
    g_2(x_k, u_k) &= a_4 x_{1,k} + a_5 x_{2,k} + b_2 u_k \\
    &\quad + \frac{|a_6|}{4}\big(x_{2,k} - \text{sgn}(a_6)u_k\big)^2, \\
    h_2(x_k, u_k) &= \frac{|a_6|}{4}\big(x_{2,k} + \text{sgn}(a_6)u_k\big)^2,
\end{aligned}
\end{equation}

Recall that the true parameter vector $\theta = [a_1, a_2, a_3, a_4, a_5, a_6, b_1, b_2]^T$ is unknown. Therefore, to evaluate the robust invariance condition in Theorem \ref{theo}, the parameters in (\ref{eq:decompo}) are replaced by the corresponding elements of each consistency set vertex $\theta^{*,s}$ for $s \in \mathbb{N}_{n_q}$. By substituting these vertex parameters into the decomposition and constructing the parameterized function $F(\cdot, \cdot, \cdot; \theta^{*,s})$ for every calculated vertex, we are able to use Algorithm \ref{alg:data_driven_compute_set} to find the largest possible polytope within $X$.

\begin{figure}[t]
    \centering
    \includegraphics[width=0.4\columnwidth]{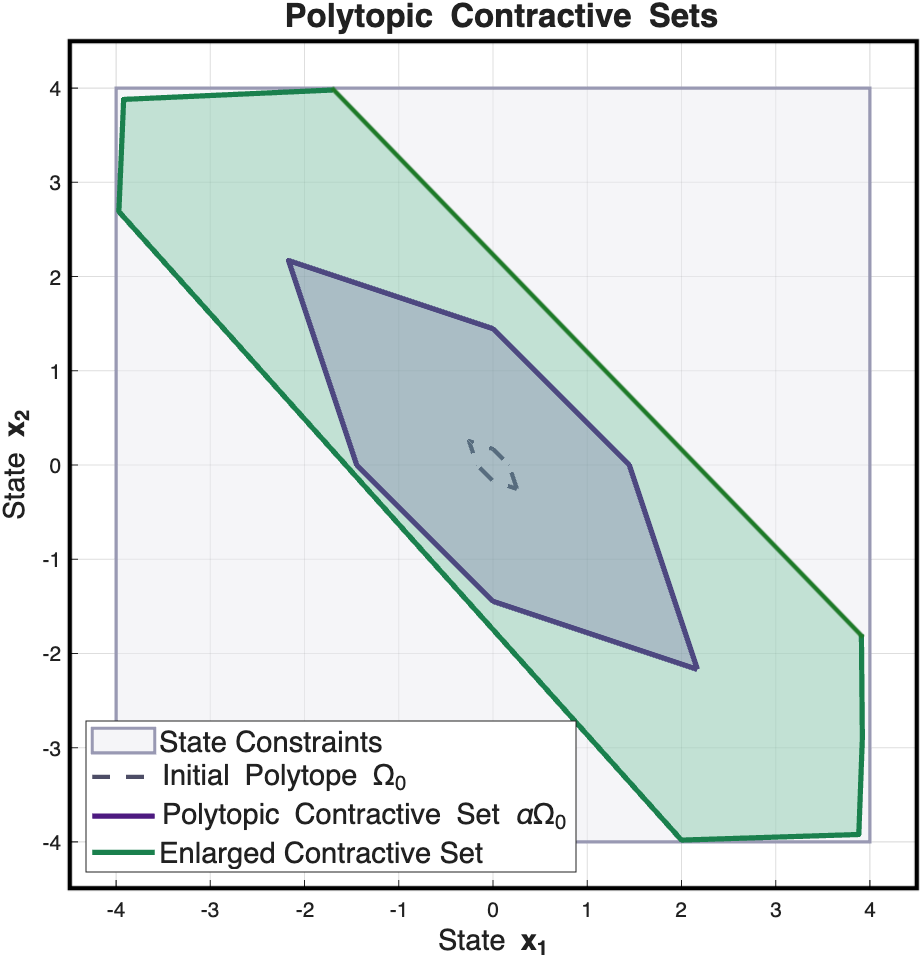}
    \caption{Data-driven computation of largest contractive set (purple solid line) within the state constraint set.}
    \label{fig:poly} 
\end{figure}

Figure \ref{fig:poly} illustrates the computational results of Algorithm \ref{alg:data_driven_compute_set}. The outer box represents the state constraint set, defined by $\|x\|_{\infty} \le 4$. The small, dashed-line ellipse at the center is the initial candidate polytope, $\Omega_0$, that provides the initial shape for the algorithm. The large, purple shaded region is the final robustly $\lambda_\omega$-contractive set computed by our data-driven method. This set is the result of scaling the initial polytope $\Omega_0$ by the factor $\alpha = 8.57$, which was determined by the algorithm to be the largest possible scaling that guarantees robust invariance for all systems compatible with available information. Although the data-driven $\lambda_\omega$-contractive set computed in Fig. \ref{fig:poly} is notably large, it remains a conservative inner approximation of the true maximal invariant set. This conservatism can be systematically reduced by employing the iterative enlarging technique detailed in Remark \ref{conv}. The green enlarged polytope presents the expanded set obtained through this procedure. Specifically, the introduction of 20 random candidate points within the allowable state constraints allows the algorithm to augment the original polytope, resulting in a demonstrably less conservative $\lambda_\omega$-contractive set.

\section{Conclusion}\label{sec:conclusion}
This paper introduced a novel data-driven framework for the approximation of robust polytopic $\lambda_\omega$-contractive sets in uncertain nonlinear systems. The proposed methodology overcomes the limitations of traditional model-based techniques by synthesizing provably invariant sets directly from a finite sequence of noisy state-input measurements. By leveraging the properties of DC functions, we successfully extended established model-based sufficient conditions for set invariance to the data-driven domain. A numerical case study validated the tractability and efficacy of the proposed optimization-based algorithm, demonstrating its potential for robust control applications.

\section{Appendix}

\begin{proof}[Proof of Theorem \ref{theo}]
Under Assumption \ref{assump4}, the consistency set is a compact polytope. This implies that for any $\theta$ within (\ref{eq:consisten}), we can express $\theta$ as a convex combination of the set's vertices, such that $\theta = \sum_{s=1}^{n_q} \beta_s\theta^{*,s}$, with coefficients $\beta_s \ge 0$ and $\sum_{s=1}^{n_q} \beta_s = 1$. Therefore, each vertex of the consistency set is interpreted as a "worst-case" realization of the unknown system parameters. Now, the dynamics associated with the parameter vector $\theta$ within the consistency set can be written as follows:

\begin{equation}
\begin{aligned}
    f(x_k,u_k;\theta) =& \sum_{s=1}^{n_q} \beta_sg(x_k,u_k;\theta^{*,s})-\sum_{s=1}^{n_q}\beta_sh(x_k,u_k;\theta^{*,s}),\\
\end{aligned}\label{decomp}
\end{equation}
where $g,h$ are the decomposed convex vector functions of the system dynamics associated with the $\theta^{*,s}$ vertex. Let (\ref{decomp}) be the decomposition for system corresponding to $\theta^{*,s}$ vertex in consistency set and the definition of $F(x, u, c; \theta^{*,s})$ proposed in (\ref{eq:F_theta}), we can write: 

\begin{equation}
\begin{aligned}
    F(x,u,c;\theta) =& \sum_{s=1}^{n_q}  \beta_s F(x,u,c;\theta^{*,s}).
\end{aligned}\label{F}
\end{equation}

Based on Property (\ref{propconv}), the function $F$ is convex. 
Considering the initial polytope $\Omega = \{x \in \mathbb{R}^n : Hx \le \mathbf{1}\} \subseteq X$, we reach at the following relationship:
\begin{equation}
F(v^p, u^p_v, H_i^T; \theta) = \sum_{s=1}^{n_q} \beta_s F(v^p, u^p_v, H_i^T; \theta^{*,s}), \label{jens}   
\end{equation}
where $v^p$ denotes the vertex of $\Omega$, and $u^p_v$ represents the input of the system at $v^p$ vertex. We can now adapt the sufficient condition for $\lambda_\omega$-contractiveness, originally proposed for known nonlinear systems in \cite{fiacchini2010computation} (Property \ref{condi}). It can be said that each vertex $\theta^{*,s}$ is associated with a worst case system parameters in consistency set and if for all $\beta_s$ in (\ref{F}), the condition (\ref{condi}) holds, the set $\Omega$ will be $\lambda_\omega$-contractive for all systems compatible with the available information.

It remains to show that there exists a control law $u$ such that for any initial state $x_0 \in \Omega$, the trajectory of the system \eqref{eq:system} subject to process noise satisfies $x_k \in \lambda_\omega^k \Omega$ for all $k \ge 0$.
To prove the existence of an admissible control law for any initial state $x \in \Omega$, we explicitly construct it as follows: Since $\Omega$ is a polytope, any $x$ can be expressed as a convex combination of its vertices $\{v^1, \dots, v^{n_v}\}$:
$$ x = \sum_{p=1}^{n_v} \zeta_p v^p, \quad \text{with } \zeta_p \ge 0, \sum_{p=1}^{n_v} \zeta_p = 1, $$
Let $u^p_v$ be the control actions at the vertices that satisfy the condition in Theorem \ref{theo}. We define the control law $u \in U$ using the same convex combination coefficients $\zeta_p$:
$$ u = \sum_{p=1}^{n_v} \zeta_p u^p_v, $$
Since the input set $U$ is convex (by Assumption \ref{assump1}), this control law is admissible for all $x \in \Omega$.

Next, we must show that for any $x_0 \in \Omega$ and any disturbance $w_0 \in W$, the successor state $x_{k+1} = f(x_k, u_k; \theta) + w_k$ is contained in $\lambda_\omega \Omega$ for all $\theta$ in the consistency set $\mathcal{P}_1$.

From the properties of the support function and the function $F$, we have the following upper bound for all $\theta$, and $\forall i \in \mathbb{N}_{n_h}$:
\begin{equation}
\begin{aligned}
H_i(f(x_k, u_k; \theta) + w_k) \le F(x_k, u_k, H_i^T; \theta) + \phi_W(H_i^T),        
\end{aligned}
\end{equation}

As established, the function $F(x,u,c;\theta)$ is convex in the pair $(x,u)$. Using (\ref{jens}) and
as shown in the proof of Theorem \ref{theo}, the condition $F(v^p, u^p_v, H_i^T; \theta) \le \lambda_\omega - \phi_W(H_i^T)$ holds for all $\theta$ in consistency set. We can write:
\begin{align*}
    H_i x_{k+1} &\le \left( \sum_{p=1}^{n_v} \zeta_p F(v^p, u^p_v, H_i^T; \theta) \right) + \phi_W(H_i^T) \\
    &\le \left( \sum_{p=1}^{n_v} \zeta_p (\lambda_\omega - \phi_W(H_i^T)) \right) + \phi_W(H_i^T) \\
    &= (\lambda_\omega - \phi_W(H_i^T)) \sum_{p=1}^{n_v} \zeta_p + \phi_W(H_i^T) \\
    &= \lambda_\omega - \phi_W(H_i^T) + \phi_W(H_i^T) = \lambda_\omega,
\end{align*}
Since $H_i x_{k+1} \le \lambda_\omega$ for all facets $i \in \mathbb{N}_{n_h}$ and vertices of consistency set, we have shown that $x_{k+1} \in \lambda_\omega \Omega$. 
\end{proof}

\bibliographystyle{IEEEtran}
\bibliography{Ref}

\end{document}